\documentclass[12pt]{article}
\newlength{\mytopmargin}
\newlength{\myleftmargin}
\usepackage{amsmath,amsfonts,amssymb,amsthm}

\usepackage{graphicx}

\newtheorem{theorem}{Theorem}
\newtheorem{proposition}[theorem]{Proposition}
\newtheorem{lemma}[theorem]{Lemma}

\begin{document}

\title{Spectral density asymptotics for Gaussian and Laguerre $\beta$-ensembles in the exponentially small region}
\author{Peter J. Forrester}
\date{}
\maketitle

\noindent
\thanks{\small Department of Mathematics and Statistics, 
The University of Melbourne,
Victoria 3010, Australia email: p.forrester@ms.unimelb.edu.au
}

\begin{abstract}
\noindent  The first two terms in the large $N$ asymptotic expansion of the $\beta$ moment of the characteristic polynomial for the
Gaussian and Laguerre $\beta$-ensembles are calculated. This is used to compute the asymptotic expansion of the spectral density
in these ensembles, in the exponentially small region outside the leading support, up to  terms $o(1)$ .
The leading form of the right tail of the distribution of the largest eigenvalue is given by the density in this regime. It is demonstrated that
there is a scaling from this, to the right tail asymptotics for the distribution of the largest eigenvalue at the soft edge.
\end{abstract}

\section{Introduction}

One of the best known results in random matrix theory is the Wigner semi-circle law (see e.g. \cite{Me91, Fo10}). In the case of the Gaussian orthogonal ensemble (${\rm GOE}_N$), which consists of $N\times N$ real symmetric matrices $H=(X+X^T)/2$ with the elements of $X$ independent standard Gaussians, this law corresponds to the limit theorem for the eigenvalue density
\begin{equation}\label{1.1}
\lim_{N\rightarrow \infty}\sqrt{\frac{2}{N}}\rho_{(1)}^{{\rm GOE}_N}(\sqrt{2N}X)=\left\{
\begin{array}{ll}
\frac{2}{\pi}(1-X^2)^{1/2}, & |X|<1\\
0, & |X|>1.
\end{array}
\right.
\end{equation}

Lesser known is the asymptotic form of $\rho_{(1)}^{{\rm GOE}_N}(\sqrt{2N}X)$ in the region $|X|>1$, where its scaled limit vanishes. Relatively recently \cite{Fy04} (see also \cite[Section 1.5]{Fo10}) this problem arose in the study of the expected number $\mathcal{N}$ of critical points for the random energy
$$
{\mathcal H} =\frac{\mu}{2}\sum_{j=1}^{N}x_j^2+V(x_1,\ldots,x_N)
$$
where $\mu>0$ and $V$ is Gaussian distributed with zero mean and covariance
\begin{equation}\label{2}
\langle V(\mathbf{x}_1)V(\mathbf{x}_2)\rangle =Nf\Big(\frac{1}{2N} (\mathbf{x}_1-\mathbf{x}_2)^2 \Big).
\end{equation}
The function $f$ in (\ref{2}) is arbitrary apart from the requirement that $J:=\sqrt{f''(0)}$ is well defined, and in fact $\mathcal{N}$ only depends on $f$ through $J$. Thus it was shown that $\mathcal{N}$ is equal to the product of terms smooth in $\mu/J$ times
$$
\int_{-\infty}^\infty e^{-N(t-(\mu/J)^2)/4}\rho_{(1)}^{{\rm GOE}_{N+1}}(\sqrt{N/2}((\mu/J)+t)) \,dt.
$$
This integral exhibits two distinct behaviours depending on whether $\mu/J<1$ or $\mu/J>1$. While knowledge of (\ref{1.1}) is sufficient to determine the former, to analyze the latter requires the $N\rightarrow \infty$ asymptotics in the region $s>1$. The density is then exponentially small, with its explicit form
derived in  \cite{Fy04} as
\begin{equation}\label{2.1}
\rho_{(1)}^{{\rm GOE}_{N+1}}(\sqrt{2N}s)\sim \exp\Big(-N(s\sqrt{s^2-1}+\log(s-\sqrt{s^2-1}))+{\rm O}(1) \Big).
\end{equation}

The result (\ref{2.1}) immediately raises two questions: that of deducing the explicit form of the ${\rm O}(1)$ term, and that of deriving the analogous asymptotics for other random matrix ensembles.
It is the aim of the present paper to answer these questions. Motivation for pursuing this line of research comes from the relevance of (\ref{2.1}) to the right tail large deviation form of the probability density function (PDF) $p_N^{{\rm GOE}}(s)$ for the largest eigenvalue in ${\rm GOE}_{N}$.

Thus, according to \cite[eq. (14.136)]{Fo10} we have that for $s>1$
\begin{equation}\label{3.2}
p_N^{{\rm GOE}}(\sqrt{2N}s)\mathop{\sim}_{N\rightarrow \infty}\rho_{(1)}^{{\rm GOE}}(\sqrt{2N}s).
\end{equation}
To derive (\ref{3.2}) we first note that
\begin{equation}\label{3.0}
p_{N}^{{\rm GOE}}(s)=-\frac{d}{ds}E_N^{{\rm GOE}}(0;(s,\infty))
\end{equation}
where $E_N^{{\rm GOE}}(0;(s,\infty))$ is the gap probability of no eigenvalues in the interval $(s,\infty)$ of the ${\rm GOE}$. But \cite[eq. (8.73)]{Fo10} the gap probability can be expressed as a series in terms of the $k$-point correlation functions $k=1,2,\dots$  according to
\begin{align}
E_N^{{\rm GOE}}(0;(\sqrt{2N}s,\infty))=&1-\sqrt{2N}\int_{s}^\infty \rho_{(1)}^{{\rm GOE}}(\sqrt{2N}t)dt \notag\\
&+\frac{(\sqrt{2N})^2}{2!}\int_{s}^\infty dt_1 \int_{s}^\infty dt_2  \, \rho_{(2)}^{{\rm GOE}}(\sqrt{2N}t_1,\sqrt{2N}t_2)+\ldots \label{3.1}
\end{align}
Assuming now that $s>1$ and noting that by asymptotic independence we must have, for $t_1,t_2,\ldots, t_k$ distinct,
\begin{equation}\label{3.1e}
\rho_{(k)}^{{\rm GOE}}(\sqrt{2N}t_1,\ldots, \sqrt{2N}t_k)\sim\ \prod_{l=1}^k \rho_{(1)}^{{\rm GOE}}(\sqrt{2N}t_l),
\end{equation}
the asymptotic form (\ref{2.1}) tells us that successive terms in (\ref{3.1}) are exponentially smaller that previous terms. Consequently, after differentiating as required by (\ref{3.0}), (\ref{3.2}) follows. Substituting (\ref{2.1}) in (\ref{3.2}) gives the leading large $N$ form of $p_N^{{\rm GOE}}(\sqrt{2N}s)$ for $s>1$. The latter was first obtained in \cite{MV09} via a Coulomb gas analysis. As an application of our determination of the analogue of (\ref{2.1}), and its extension by the explicit determination of the ${\rm O}(1)$ term, for other random matrix ensembles we can use the analogue of (\ref{3.2}) to deduce a corresponding large deviations asymptotic form of the right tail of the distribution of the largest eigenvalue in the ensemble.

The ensembles to be considered are the Gaussian and Laguerre $\beta$-ensembles. They can be constructed out of certain tridiagonal and bidiagonal random matrices with independently distributed elements \cite{ES06}. These matrices in turn, for the special $\beta$ values $1,2$ and $4$, are Householder similarity reductions of the Gaussian matrices used to construct the classical Gaussian and Laguerre orthogonal, unitary and symplectic ensembles. For present purposes the $\beta$-ensembles can be defined by their eigenvalue PDF. Generally, let us denote by ${\rm ME}_{\beta,N}(w)$ the eigenvalue PDF proportional to
\begin{equation}\label{3.a}
\prod_{l=1}^N w(\lambda_l)\prod_{1\leq j<k\leq N}|\lambda_k-\lambda_j|^\beta.
\end{equation}
Then the Gaussian and Laguerre $\beta$-ensembles correspond to the eigenvalue PDFs
$$
{\rm ME}_{\beta,N}(e^{-\beta\lambda^2/2}) \text{ and } {\rm ME}_{\beta,N}(\lambda^{\beta a /2}e^{-\lambda}\chi_{\lambda>0})
$$
respectively, where $\chi_A=1$ for $A$ true and $\chi_A=0$ otherwise.

Our aim is to compute the leading two terms in the asymptotic expansion of the density for the Gaussian and Laguerre $\beta$-ensembles, in the large deviation regime outside the leading order support. This will be done by expressing the density in terms of a particular moment of the characteristic polynomial. The latter can be interpreted in terms of the characteristic function of a linear statistic, for which there is a known asymptotic formula. This then leads to the sought asymptotic expansion of the density.

In Section \ref{Section2} we give the details of the relation between the density and a moment of the corresponding characteristic polynomial, and we revise how this relates to the characteristic function of a linear statistic. The large $N$ form of a large class of linear statistics for PDFs ${\rm ME}_{\beta, N}(w)$ is expected to exhibit a Gaussian form, with explicit integral expressions for the mean and variance. The latter are evaluated in Sections \ref{Section3} and \ref{Section4} for the Gaussian and Laguerre $\beta$-ensembles respectively, and the leading two terms in the asymptotic expansion of the corresponding densities outside the leading order support are given in those sections also. A discussion of consequences of these results with regards to the asymptotic form of the right tail of the distribution of the largest eigenvalue is given in Section \ref{Section5}, as is a possible experimental realization in the case $\beta = 2$.

Although the Gaussian expression for the asymptotic form of the characteristic function is well founded, its rigorous proof in cases
corresponding to the moments of characteristic polynomials is known only for ensembles with two soft edge boundaries,
or a hard edge with parameter $a=0$
\cite{BG11}. This covers the Gaussian $\beta$-ensemble, but only the $a=0$ Laguerre $\beta$-ensemble.
On the other hand, it is known \cite{BF97a} that for $\beta$ even these moments can be expressed as $\beta$-dimensional integrals. In the Appendix we show that the rigorous large $N$ expansion of the latter in both the Gaussian and Laguerre cases agrees with our earlier working.

\section{The density as a moment of the characteristic polynomial}\label{Section2}
\setcounter{equation}{0}
\subsection{Exact expressions}

For the matrix ensemble ${\rm ME}_{\beta,N}(w)$ the eigenvalue density $\rho_{(1),n}(x)$ is given in terms of a multiple integral by
\begin{equation}\label{6.1}
\rho_{(1),n}(x)=\frac{N}{C_{\beta,N}[w]}w(x)\int_{-\infty}^\infty d \lambda_2  \, w(\lambda_2)\ldots \int_{-\infty}^\infty d\lambda_N
\, w(\lambda_N)\prod_{1\leq j<k\leq N}|\lambda_k-\lambda_j|^\beta,
\end{equation}
where
$$
C_{\beta,N}[w]=\int_{-\infty}^\infty d\lambda_1 \, w(\lambda_1)\ldots \int_{-\infty}^\infty d\lambda_N \, w(\lambda_N)\prod_{1\leq j<k\leq N}|\lambda_k-\lambda_j|^\beta.
$$
Defining for general $g$ the canonical average with respect to ${\rm ME}_{\beta,N}(w)$ by the multiple integral
\begin{align}
\langle g(\lambda_1,\ldots, \lambda_N) \rangle_{{\rm ME}_{\beta,N}(w)}:=&\frac{1}{C_{\beta,N}[w]}\int_{-\infty}^\infty \, d\lambda_1 w(\lambda_1)\ldots \int_{-\infty}^\infty d\lambda_N \, w(\lambda_N) \notag \\
&\times \prod_{1\leq j<k\leq N}|\lambda_k-\lambda_j|^\beta g(\lambda_1,\ldots, \lambda_N),\label{6.2}
\end{align}
we have that
\begin{equation}\label{6.3}
\Big\langle \prod_{l=1}^N |x-\lambda_l|^\mu\Big\rangle_{{\rm ME}_{\beta,N}(w)}
\end{equation}
corresponds to the $\mu$-th moment of the modulus of the characteristic polynomial of the matrix ensemble.

It is immediate from the definitions (\ref{6.1}) and (\ref{6.2}) that the eigenvalue density can be written in terms of a moment of the characteristic polynomial,
\begin{equation}\label{d2}
\rho_{(1),N+1}(x)=\frac{(N+1)C_{\beta,N}[w]}{C_{\beta,N+1}[w]}w(x)\Big\langle \prod_{l=1}^N|x-\lambda_l|^\beta \Big\rangle_{{\rm ME}_{\beta,N}(w)}.
\end{equation}
The corresponding inter-relation in the case of the real Ginibre matrices ($N\times N$ matrices with entries independent standard Gaussians) was noted by Edelman, Kostlan and Shub \cite{EKS94}, and applied for purposes of computing the density and expected number of real eigenvalues.

For the Gaussian and Laguerre $\beta$-ensembles we want to scale the eigenvalues so that for $N\rightarrow \infty$ the support of the density is a finite interval. For this purpose let $M$ be such that $M/N\rightarrow 1$. Then it is well known \cite{Fo93a} that for the Gaussian $\beta$-ensemble with $\lambda_l\mapsto \sqrt{2M}\lambda_l$, and the Laguerre $\beta$-ensemble with $\lambda_l\mapsto 4M\lambda_l$, the $N\rightarrow \infty$ support is the intervals $(-1,1)$ and $(0,1)$ respectively. In these scaled coordinates $|\lambda|>1$ will then correspond to an exponentially small density for large $N$.

Introducing these scalings, for the Gaussian $\beta$-ensemble we then have
\begin{equation}\label{d.3a}
\sqrt{2M}\rho_{(1),N+1}^{\rm G}(\sqrt{2M}x)=A^{\rm G}e^{-\beta Mx^2}\Big\langle \prod_{l=1}^N|x-\lambda_l|^\beta \Big\rangle_{{\rm ME}_{\beta,N}(e^{-\beta M \lambda^2})}
\end{equation}
where
$$
A^{\rm G}=(N+1)(2M)^{(N\beta+1)/2}\frac{C_{\beta,N}[e^{-\beta \lambda^2/2}]}{C_{\beta,N+1}[e^{-\beta \lambda^2/2}]},
$$
while for the Laguerre $\beta$-ensemble
\begin{equation}\label{d.3b}
4M\rho_{(1),N+1}^{\rm L}(4Mx)=A^{\rm L}x^{a\beta/2}e^{-2M\beta x}\Big\langle \prod_{l=1}^N|x-\lambda_l|^\beta \Big\rangle_{{\rm ME}_{\beta,N}(\lambda^{a\beta/2}e^{-2\beta M \lambda})}
\end{equation}
where
$$
A^{\rm L}=(N+1)(4M)^{N\beta+1+a\beta/2}\frac{C_{\beta,N}[\lambda^{a\beta/2}e^{-\beta \lambda/2}]}{C_{\beta,N+1}[\lambda^{a\beta/2}e^{-\beta \lambda/2}]}.
$$
The normalizations in (\ref{d.3a}), (\ref{d.3b}) are particular Selberg integrals (see \cite[Ch. 4]{Fo10}) and as such have evaluations in terms of products of gamma functions. This provides us with the explicit expressions
\begin{align}
A^{\rm G}&=\frac{(N+1)}{(2\pi)^{1/2}}(2M\beta)^{(N\beta+1)/2}\frac{\Gamma(1+\beta/2)}{\Gamma(1+(N+1)\beta/2)}, \label{9a}\\
A^{\rm L}&=(N+1)(2M\beta)^{N\beta+1+a \beta/2}\frac{\Gamma(1+\beta/2)}{\Gamma(1+(N+1)\beta/2)\Gamma(a\beta/2+1+N\beta/2)}. \label{9b}
\end{align}

For future purposes we note that in the case $M=N+1$ Stirling's formula applied to (\ref{9a}) and (\ref{9b}) shows
\begin{align}
A^{\rm G} &\sim \frac{N}{\pi}(N\beta)^{-\beta/2}2^{(N+1/2)\beta}e^{(N+1)\beta/2}\Gamma(1+\beta/2),  \label{9ax}\\
A^{\rm L} &\sim \frac{N}{\pi}(N\beta)^{-\beta/2}2^{2 N \beta + 1 + \beta/2 + a \beta}e^{(N+1)\beta }\Gamma(1+\beta/2).  \label{9ag}
\end{align}

\subsection{Asymptotic form of the averages}

Consider the linear statistic
\begin{equation}\label{V1}
V(x)= \sum_{l=1}^N\log |x-\lambda_l|,
\end{equation}
where $\lambda_l$ are chosen with PDF (\ref{3.1}). The distribution of this linear statistic is given by the canonical average
$$
P_V(t)=\Big\langle \delta(t-\sum_{l=1}^N\log|x-\lambda_l|) \Big\rangle_{{\rm ME}_{\beta,N}(w)}.
$$
Taking the Fourier transform of both sides gives the corresponding characteristic function
\begin{equation}\label{7.e}
\widehat{P}_V(k) = \Big\langle e^{ik\sum_{l=1}^N\log|x-\lambda_l|}\Big\rangle_{{\rm ME}_{\beta,N}(w)}.
\end{equation}
Observing 
$$
\Big\langle \prod_{l=1}^N |x-\lambda_l|^\beta \Big\rangle_{{\rm ME}_{\beta,N}(w)}=\Big\langle e^{\beta\sum_{l=1}^N \log|x-\lambda_l|} \Big\rangle_{{\rm ME}_{\beta,N}(w)}
$$
we thus have
\begin{equation}\label{11.1}
\Big\langle \prod_{l=1}^N|x-\lambda_l|^\beta \Big\rangle_{{\rm ME}_{\beta,N}(w)}=\widehat{P}_V(-i \beta).
\end{equation}

More generally let $G:=\sum_{j=1}^N g(\lambda_j)$ denote an arbitrary linear statistic.
The characteristic function is then
\begin{equation}\label{2.13a}
\widehat{P}_G(k) =\Big\langle \prod_{l=1}^N e^{ik \sum_{j=1}^N g(\lambda_j)} \Big\rangle_{{\rm ME}_{\beta,N}(w)}.
\end{equation}
The significance of the relation (\ref{11.1}) is that
for a large class weights $w$ and statistics $g$, (\ref{2.13a}) has a known asymptotic form.
Thus suppose that for large $N$ the eigenvalue support of ${\rm ME}_{\beta,N}(g)$ is the finite interval $[a,b]$. Suppose too that $g(\lambda)$ and its derivative are bounded on $[a,b]$ or more generally that (\ref{m2}) below is finite. In these circumstances there are well founded physical grounds (see e.g. \cite[Section 14.4.1]{Fo10}) to expect that $\widehat{P}_G(k)$ is a Gaussian,
\begin{equation}\label{ea2}
\widehat{P}_G(k)\mathop{\sim}_{N\rightarrow \infty} e^{ik \mu_N(g)}e^{-k^2(\sigma(g))^2/2},
\end{equation}
where the mean $\mu_N(g)$ and variance $(\sigma(g))^2$ have the explicit forms
\begin{align}
\mu_N(g)=&\int_a^b\rho_{(1),N}(t)g(t) \, dt \label{m1}\\
(\sigma(g))^2 =&\frac{1}{\beta \pi^2}\int^b_a dt_1 \, \frac{g(t_1)}{((b-t_1)(t_1-a))^{1/2}}\int_a^b dt_2 \, \frac{g'(t_2)((b-t_2)(t_2-a))^{1/2}}{t_1-t_2}\notag \\
=&\frac{1}{2\beta}\sum_{k=1}^\infty ka_k^2,\hspace{1cm}a_k:=\frac{2}{\pi}\int_0^\pi g\Big(\frac{a+b}{2}+\frac{b-a}{2}\cos\theta \Big)\cos k\theta \, d\theta \label{m2}
\end{align}
(see also \cite{FF03}). At a rigorous level, for a class of weights $w(x)$ including the Gaussian but not Laguerre, and a class of $g(x)$ which does not include $g(x)= \log|x-t|$, this has been proved by Johansson \cite{Jo98}. Very recently \cite{BG11}, for matrix ensembles in which the eigenvalue support is
a single interval, with both endpoints soft edges or a hard edge with parameter $a=0$, the conditions on $g(x)$ have been relaxed to require only that $g(x)$ be analytic in a neighbourhood of
the support. The Gaussian $\beta$-ensemble has the first of these properties, with  $g(x)= \log|x-t|$ for $x$ outside the eigenvalue support has the second
of the properties, and thus then (\ref{ea2}) is a rigorous theorem. But in general the Laguerre ensemble has a hard edge with $a \ne 0$, and so the result of
 \cite{BG11} applies only in the special case $a=0$.

We will proceed under the assumption that (\ref{ea2}) is valid for the linear statistic (\ref{V1}) not only in the Gaussian case, but also
in the general  Laguerre case,
with $x$ outside the interval of support of the eigenvalue density $[a,b]$. Then use of (\ref{ea2}) in (\ref{11.1}) gives
\begin{equation}\label{9.1}
\left\langle \prod_{l=1}^N|x-\lambda_l|^\beta \right\rangle_{{\rm ME}_{\beta,N}(w)}\sim e^{\beta \mu_N(v)}e^{(\beta \sigma(v))^2/2}
\end{equation}
with $v(t)= \log|x-t|$ and appropriate $w$. It thus remains to evaluate $\mu_N(v)$ and $\sigma(v)$. Substitution in (\ref{d.3a}) will then give the sought large $N$ form of the densities. We will treat the Gaussian and Laguerre cases separately.

\section{The Gaussian $\beta$-ensemble}\label{Section3}
\setcounter{equation}{0}
For the weight $w(\lambda)=e^{-M\beta \lambda^2}$ in (\ref{3.a}) with $M/N\rightarrow 1$ as $N\rightarrow \infty$ the leading form of the density $\rho_{(1),N}(t)$ is given by the RHS of (\ref{1.1}) multiplied by $N$. The next order term consists of an oscillatory and a non-oscillatory part \cite{DF06} --- only in the latter contributes to the integral in (\ref{m1}) to the same order. Ignoring then the oscillatory term we read off from Johansson \cite[displayed equation below (2.15)]{Jo98} or more explicitly \cite[eq.~(19), Lemma 2.20]{DE05}, \cite[eq. (1.4)]{FFG06e} that for $|t|\leq 1$
\begin{align}
\rho_{(1),N}(t)\sim& \frac{2(MN)^{1/2}}{\pi}\Big(1-\frac{M}{N}t^2 \Big)^{1/2}
+\Big(\frac{1}{\beta}-\frac{1}{2} \Big)\Big(\frac{1}{2}\delta(t-1)+\frac{1}{2}\delta(t+1)-\frac{1}{\pi}\frac{1}{\sqrt{1-t^2}} \Big)\notag\\
\sim & \frac{2M}{\pi}(1-t^2)^{1/2} 
+\Big(\frac{1}{\beta}-\frac{1}{2} \Big)\Big(\frac{1}{2}\delta(t-1)+\frac{1}{2}\delta(t+1)-\frac{1}{\pi}\frac{1}{\sqrt{1-t^2}}\Big)\notag\\
&\;\;+\frac{N-M}{\pi\sqrt{1-t^2}},\label{w1}
\end{align}
where the second asymptotic equality follows by expanding $M/N$ about $1$.

We seek the value of (\ref{m1}) with $\rho_{(1),N}(t)$ given by (\ref{w1}) and $a(t)=\log(x-t)$. For this we make use of the following integral evaluations.

\begin{proposition}\label{Wx}
Suppose $x>1$. We have
\begin{align}
&\frac{2}{\pi}\int_{-1}^1\log|x-t| \, (1-t^2)^{1/2}dt=x^2-x\sqrt{x^2-1}-\log(2(x-\sqrt{x^2-1}))-\frac{1}{2}, \label{r1}\\
&\frac{1}{\pi}\int_{-1}^1\log|x-t| \, (1-t^2)^{-1/2}dt=\log\Big(\frac{1}{2}(x+\sqrt{x^2-1})\Big).\label{r2}
\end{align}
\end{proposition}
\begin{proof}
It is straightforward to check that for both equations, both sides tend to $\log x +{\rm O}(1/x)$ as $x\rightarrow \infty$. Hence it suffices to show that both sides have the same derivative with respect to $x$, and thus that for $x>1$
\begin{align*}
&\frac{2}{\pi}\int_{-1}^1\frac{(1-t^2)^{1/2}}{x-t}dt=2x(1-(1-1/x^2)^{1/2}),\\
&\frac{1}{\pi}\int^1_{-1}\frac{1}{(1-t^2)^{1/2}}\frac{dt}{x-t}=\frac{1}{(x^2-1)^{1/2}}.
\end{align*}
The first is a well known result in random matrix theory (see e.g. \cite[eq. (1.32)]{Fo10}). Both integral evaluations can be checked by expanding the left hand and right hand sides in powers of $1/x$, and making use of the beta integral on the LHS and the binomial expansion on the RHS.
\end{proof}

Making use of Proposition \ref{Wx} we therefore have
\begin{align}
\mu_N=&M\Big(x^2-x\sqrt{x^2-1}-\log(2(x-\sqrt{x^2-1}))-\frac{1}{2} \Big)\notag\\
&+\Big(\frac{1}{\beta}-\frac{1}{2} \Big)\Big(\frac{1}{2}\log|1-x^2|-\log\Big(\frac{1}{2}(x+\sqrt{x^2-1})\Big) \Big)\notag\\
&+(N-M)\log\Big(\frac{1}{2}(x+\sqrt{x^2-1}) \Big).\label{r1a}
\end{align}
Next we are next faced with the task of evaluating (\ref{m2}) for $a(t)=\log|x-t|$ and $a=-1,b=1$. Making use of the second equality in (\ref{m2}), simple manipulation shows we must evaluate
\begin{equation}\label{Ld}
a_k=\frac{2}{\pi}\int_{0}^\pi \log\Big(1-\frac{\cos \theta}{x} \Big) \cos k\theta \, d\theta
\end{equation}
for $k=1,2,\ldots $ and $x>1$.

\begin{lemma}\label{Lem1}
For $x>1$
\begin{equation}\label{Ls}
\log\Big(1-\frac{\cos \theta}{x} \Big)=\log(1-\nu e^{i\theta})(1-\nu e^{-i\theta})-\log(1+\nu^2),
\end{equation}
where
\begin{equation}\label{Ls1}
\nu=x-(x^2-1)^{1/2}
\end{equation}
and thus in particular $0<\nu<1$.
\end{lemma}

\begin{proof}
The RHS of (\ref{Ls}) is equal to $$
\log\Big(1-\frac{2\nu}{1+\nu^2}\cos\theta \Big).
$$
Equating with the LHS we read off that $$
\frac{1}{x}=\frac{2\nu}{1+\nu^2},
$$
and solving for $\nu$ gives (\ref{Ls1}).
\end{proof}

Substituting (\ref{Ls}) in (\ref{Ld}) shows
\begin{equation}\label{Ls2}
a_k=-\frac{2\nu^k}{k},
\end{equation}
and substituting this in the second equality of (\ref{m2}) we see upon evaluating the sum that
\begin{align}
\sigma^2&=-\frac{2}{\beta}\log(1-\nu^2)\notag\\
&=-\frac{2}{\beta}\log(x^2-1)^{1/2}-\frac{2}{\beta}\log(2(x-\sqrt{x^2-1})),\label{r2a}
\end{align}
where the second equality follows upon making use of (\ref{Ls1}).

Let us now make the choice $M=N+1$. Then it follows from (\ref{9.1}), (\ref{r1a}) and (\ref{r2a}) that
\begin{align}
\left\langle \prod_{l=1}^N|x-\lambda_l|^\beta\right\rangle_{{\rm ME}_{\beta,N}(e^{-(N+1)\beta\lambda^2})}
&\mathop{\sim}_{N\rightarrow \infty}e^{(N+1)\beta(x^2-x\sqrt{x^2-1}-\log(2(x-\sqrt{x^2-1}))-1/2)}\notag\\
&\hspace{1cm}\times\frac{(x^2-1)^{(1-3\beta/2)/2}}{((x+(x^2-1)^{1/2})/2)^{(1-\beta/2)}}.\label{wf}
\end{align}
We remark that with the exponent in the average on the LHS replaced by $2 \alpha$, and $\beta$ in ME${}_{\beta,N}$ set
equal to 2, the asymptotic formula (\ref{ea2}) has been used in \cite{K07} to deduce the corresponding asymptotic form
up to the same order as in (\ref{wf}).
With $\alpha = 1$ the result of  \cite{K07} agrees with the $\beta = 2$ case of (\ref{wf}).

The result (\ref{wf}) substituted in (\ref{d.3a}) and use of (\ref{9ax}) then gives our sought generalization of (\ref{2.1}),
\begin{align}
\sqrt{2N}\rho_{(1),N}^{\rm G}(\sqrt{2N}x)\mathop{\sim}_{N\rightarrow \infty}&e^{-N\beta(x\sqrt{x^2-1}+\log(x-\sqrt{x^2-1}))}\notag\\
&\times \frac{(x^2-1)^{(1-3\beta/2)/2}}{((x+(x^2-1)^{1/2})/2)^{(1-\beta/2)}}\left( \frac{N}{\pi}(N\beta)^{-\beta/2}2^{1-\beta}\Gamma(1+\beta/2)\right)\label{wf1}
\end{align}
valid for $x>1$.
As a check on our working we note that in the case $\beta=2$, corresponding to the GUE, the leading two terms in the large $N$ expansion of the largest eigenvalue distribution $p_N^{{\rm GUE}}(\sqrt{2N}s),$ $s>1$, has been obtained recently \cite[eq. (76) with $\alpha=1$ and $t=\sqrt{2N}s$]{NM11}. But the analogue of (\ref{3.2}) for the GUE tells us that this asymptotic form must be identical to that for $\rho_{(1),N}^{\rm G}(\sqrt{2N}s)|_{\beta=2}$. Indeed (\ref{wf1}) with $\beta=2$ reproduces the result from \cite{NM11}.

In the cases $\beta = 1,2$ and 4 simple closed form expressions for $\rho_{{(1)},N}^{\rm G}(x)$ are known in terms of Hermite polynomials
(see e.g.~\cite[Ch.~5\&6]{Fo10}). The simplest case is $\beta = 2$ for which
\begin{equation}\label{17.1}
\rho_{(1),N}^{\rm G}(x) = {2^{-N} e^{-x^2} \over \sqrt{\pi} (N-1)!} \Big ( H_N'(x) H_{N-1}(x) -  H_{N-1}'(x) H_{N}(x) \Big ),
\end{equation}
while for $\beta = 1$ and $N$ even we have
\begin{equation}\label{17.2}
\rho_{(1),N}^{\rm G}(x) = \rho_{(1),N-1}^{\rm G}(x) \Big |_{\beta = 2} +
{e^{-x^2/2} \over \sqrt{\pi} 2^{(N-1)} (N-2)! }H_{N-1}(x) {1 \over 2}
\int_{-\infty}^\infty {\rm sgn} (x-t) H_{N-2}(t) e^{-t^2/2} \, dt.
\end{equation}
The case $\beta = 4$ is similar to (\ref{17.2}), but to make our point is suffices to restrict attention to $\beta = 1$ and 2.

Numerical evaluation of (\ref{17.1}) and (\ref{17.2}) allows the accuracy of (\ref{wf1}) to be quantified.
Define the ratio
\begin{equation}\label{r1}
r^{\rm G}_{\beta,N}(s) := {\rho_{(1),N}^{\rm G}(\sqrt{2N}s) \over \rho_{(1),N}^{\rm G, asym}(\sqrt{2N}s) }
\end{equation}
where $ \rho_{(1),N}^{\rm G, asym}(\sqrt{2N}s)$ is the asymptotic form (\ref{wf1}). The results for particular $s$ and varying $N$
are given in Table \ref{T1a}.

\begin{table}
\begin{center}
\begin{tabular}{c||c|c|c|c||c|c|c|c|}
&1.2 & 1.4 & 1.6 & 1.8 & 1.2 & 1.4 & 1.6 & 1.8 \\\hline
6 &1.404 & 1.177& 1.109  & 1.078 &1.126&1.074&1.056&1.047 \\
12&1.220 & 1.092 & 1.055 & 1.039 &1.067&1.038&1.028&1.023 \\
18 & 1.152&1.062 & 1.037 & 1.026&1.046&1.025&1.019&1.015\\
24 & 1.116&1.047&1.028&1.019& 1.035 & 1.019 & 1.014 & 1.011\\
30&1.094&1.038&1.022&1.015 & 1.028 & 1.015 & 1.011 & 1.009
\end{tabular}
\caption{\label{T1a} The ratios $r^{\rm G}_{2,N}(s)$ (first four inner columns) and
$r^{\rm G}_{1,N}(s)$ (final four columns) for varying $N$ (rows) and given $s$ (columns).}
\end{center}
\end{table}

\section{The Laguerre $\beta$-ensemble}\label{Section4}
\setcounter{equation}{0}

We now turn our attention to the weight $w(\lambda)=\lambda^{a\beta/2}e^{-2M\beta \lambda}\chi_{\lambda>0}$ in (\ref{3.a}) with $M/N\rightarrow 1$ as $N\rightarrow \infty$. The leading two terms in the large $N$ expansion of the smoothed density for $0<t<1$ is most naturally presented in the variable $t^2$. Thus according to \cite[eqns. (6.21) and (6.22)]{FFG06e} we have
\begin{align}
2t\rho_{(1),N}(t^2)\sim&\frac{4(MN)^{1/2}}{\pi}\Big(1-\frac{M}{N}t^2 \Big)^{1/2}
+\Big(a+\frac{1}{\beta}-\frac{1}{2}\Big)\Big(\frac{1}{\pi\sqrt{1-t^2}}-\frac{1}{2}\delta(t) \Big)\notag\\
&+\Big(\frac{1}{\beta}-\frac{1}{2} \Big)\Big(\frac{1}{2}\delta(t-1)-\frac{1}{\pi\sqrt{1-t^2}} \Big) \notag\\
\sim&\frac{4M}{\pi}(1-t^2)^{1/2}
+\Big(a+\frac{1}{\beta}-\frac{1}{2}\Big)\Big(\frac{1}{\pi\sqrt{1-t^2}}-\frac{1}{2}\delta(t) \Big)\notag\\
&+\Big(\frac{1}{\beta}-\frac{1}{2} \Big)\Big(\frac{1}{2}\delta(t-1)-\frac{1}{\pi \sqrt{1-t^2}} \Big)+\frac{2(N-M)}{\pi\sqrt{1-t^2}}.\label{6.e}
\end{align}
Noting that
$$
\int_0^1\log|x-t|\rho_{(1),N}(t) \, dt=\int_{-1}^1\log|\sqrt{x}-t|(2t\rho_{(1),N}(t^2)) \, dt,
$$
where we define $2t\rho_{(1),N}(t^2)$ for $t<0$ by the requirement that it be even (equivalently replace $\delta(t-1)$ in (\ref{6.e}) by $\delta(t-1)+\delta(t+1)$) we see that the corresponding form of (\ref{m1}) with $a(t)=\log|x-t|$ can be simplified using knowledge of (\ref{r1}) and (\ref{r2}). Thus, with $M=N+1$
\begin{align}
\mu_N=&2M\Big(x-\sqrt{x(x-1)}-\log(2(\sqrt{x}-\sqrt{x-1}))-\frac{1}{2}\Big)\notag \\
&+(a-2)\log\Big(\frac{1}{2}(\sqrt{x}+\sqrt{x-1}) \Big)-\frac{1}{2}(2a+2/\beta-1)\log|x|^{1/2} \notag\\
&+\Big(\frac{1}{\beta}-\frac{1}{2} \Big)\log(x-1)^{1/2}.\label{19.1}
\end{align}

Our remaining task is to evaluate (\ref{m2}) for $a(t)=\log|x-t|$ and $a=0$, $b=1$. Making use of the second equality in (\ref{m2}) shows we must evaluate
$$
a_k=\frac{2}{\pi}\int_{0}^\pi\log\Big(1-\frac{\cos\theta}{2x-1} \Big)\cos k\theta \, d\theta
$$
for $k=1,2,\ldots$ and $x>1$. For this we can use Lemma \ref{Lem1} with $x$ replaced by $2x-1$ in (\ref{Ls1}) and thus
\begin{equation}\label{Ls3}
\nu=2x-1-2x^{1/2}(x-1)^{1/2}.
\end{equation}
In particular it follows that $a_k$ is again evaluated according to (\ref{Ls2}), but with $\nu$ now given by (\ref{Ls3}). Consequently
\begin{align}
\sigma^2&=-\frac{2}{\beta}\log(1-\nu^2)\notag\\
&=-\frac{2}{\beta}\log (x(x-1))^{1/2}-\frac{4}{\beta}\log(2(x^{1/2}-(x-1)^{1/2})). \label{19.2}
\end{align}

The results (\ref{19.1}) and (\ref{19.2}) substituted in (\ref{9.1}) give for the asymptotic form of the $\beta$ moment of the characteristic polynomial for the Laguerre $\beta$-ensemble,
\begin{align}
\left\langle
\prod_{l=1}^N|x-\lambda_l|^\beta\right\rangle_{{\rm ME}_{\beta,N}(\lambda^{a\beta/2}e^{-2(N+1)\beta\lambda})}
\mathop{\sim}_{N\rightarrow \infty} & e^{2(N+1)\beta(x-\sqrt{x(x-1)}-\log(2(\sqrt{x}-\sqrt{x-1})-1/2)}\notag\\
&\times 2^{-a\beta}\frac{(x-1)^{(1-3\beta/2)/2}}{x^{((a+1/2)+1)/2}}(\sqrt{x}+\sqrt{x-1})^{a\beta}.\label{20.1}
\end{align}
We remark that for fixed $N$ and $x$ large the LHS tends to $x^{\beta N}$. It is straightforward to check that this is the $x$ large
behaviour of the RHS, suggesting that (\ref{20.1}) holds uniformly for $x > 1$ (the corresponding Gaussian asymptotic formula (\ref{wf})
also exhibits this property).
Substituting this in (\ref{d.3b}) with $M=N+1$ and making use too of (\ref{9ag}) we obtain for the asympotic form of the density
\begin{align}
4N\rho_{(1),N}^{\rm L}(4Nx)\mathop{\sim}_{N\rightarrow\infty}&e^{-2N\beta(\sqrt{x(x-1)}+\log(\sqrt{x}-\sqrt{x-1}))}\notag\\
&\times \frac{(x-1)^{(1-3\beta/2)/2}}{x^{(\beta/2+1)/2}}(\sqrt{x}+\sqrt{x-1})^{a\beta}\Big( 2^{1-3\beta/2}\Gamma\Big(1+\frac{\beta}{2}\Big)\frac{N}{\pi}(N\beta)^{-\beta/2}\Big).\label{20.2}
\end{align}

In the Laguerre case, for $\beta = 1$, 2 and 4 there are simple closed form expressions for the density in terms of Laguerre polynomials
(see e.g.~\cite[Ch.~5\&6]{Fo10}. The simplest case is $\beta = 2$, for which
$$
\rho_{(1)}^{\rm L}(x) = {\Gamma(N+1) \over \Gamma(N+a)} x^{a}e^{-x} \Big ( L_N^a(x) {d \over dx} L_{N-1}^a(x) -
L_{N-1}^a(x) {d \over dx} L_{N}^a(x) \Big ).
$$
The explicit formulas for  $\beta = 1$ and 4 has a structure similar to (\ref{17.2}). We have used these formulas (in the cases $\beta = 2$ and 1) to
check that that indeed for given $a$, a fixed $x$ and increasing $N$ the ratio of the LHS and RHS in (\ref{20.2}) tends to unity.

\section{Distribution of the largest eigenvalue}\label{Section5}
\setcounter{equation}{0}

We have noted in the case of the GOE that the large deviation form of the right tail of the PDF for the distribution of the largest eigenvalue is, up to exponentially small corrections, equal to the corresponding large deviation form of the density (recall eq. (\ref{3.2})). The given derivation of this formula applies equally as well to the right tail of the largest eigenvalue PDF for the Gaussian and Laguerre $\beta$-ensembles. Thus with $s>1$ we have
\begin{align}
\label{F1} p_N^{\rm G} (\sqrt{2N} s) \sim \rho_{(1), N}^{\rm G} (\sqrt{2N} s),
\end{align}
where the RHS has the explicit form given by (\ref{wf1}), and
\begin{align}
\label{F2} p_{N}^{\rm L} (4N s) \sim \rho_{(1), N}^{\rm L} (4N s),
\end{align}
where the RHS has the explicit form given by (\ref{20.2}).

In \cite{MV09} the leading asymptotic expressions
\begin{align}
\label{51a} p_N^{\rm G} (\sqrt{2N} s) &\sim \exp \left(-\beta N \psi_+^{\rm G} (\sqrt{2}(s-1))\right),\\
\label{51b} p_N^{\rm L} (4N s) &\sim \exp \left(-\beta N \psi_+^{\rm L }(4(s-1))\right),
\end{align}
valid for $s>1$ were derived.
Here, with $G(z):= {}_3F_2 (\{ 1,1,3/2\},\{ 2,3\}; z)$ a particular hypergeometric function, the rate functions $\psi_+^{\rm G}$ and $\psi_+^{\rm L}$ are specified by
\begin{align}
\label{52a} \psi_{+}^{\rm G} (z-\sqrt{2}) &:= \frac{z^2-1} {2} - \log (z\sqrt{2}) +\frac{1}{4z^2} \; G\left( \frac{2} {z^2}\right), \\
\label{52b} \psi_+^{\rm L} (z-4) &:= \frac{z-2}{2} -\log z +\frac{1}{z} \; G\left(\frac{4} {z} \right).
\end{align}
Subsequently in \cite[eq. (15)]{NM11} the first of these was written in the simplified form
\begin{align}
\label{53} \psi_+^{\rm G} (z-\sqrt{2}) = \frac{z\sqrt{z^2- 2}} {2} +\log \left( \frac{z-\sqrt{z^2-2}} {\sqrt{2}}\right).
\end{align}
Substituted in (\ref{51a}) this implies
\begin{align}
\label{54} p_N^{\rm G} (\sqrt{2N} s) \sim e^{-N\beta (s\sqrt{s^2-1}) +\log (s-\sqrt{s^2-1}))},
\end{align}
which indeed is the leading asymptotic behaviour as predicted by (\ref{F1}) and (\ref{wf1}). It is remarked in \cite{NM11} that (\ref{54}) in the GOE case $\beta=1$ was first established in an earlier work \cite{BDG01}.

Note that comparison of (\ref{52a}) and (\ref{53}) implies an explicit form for $G(2/z^2)$. This substituted in (\ref{52b}) implies
\begin{align}
\label{55} \psi_+^{\rm L} (z-4) = \frac{1}{2} \sqrt{z(z-4)} + 2\log \frac{1}{2} \left( \sqrt{z} - \sqrt{z-4}\right).
\end{align}
Now substituting (\ref{55}) in (\ref{51b}) we obtain the explicit large derivation formula
\begin{align}
\label{56} p_N^{\rm L} (4Ns) \sim e^{-2N\beta (\sqrt{s(s-1)} + \log (\sqrt{s}-\sqrt{s-1}))}.
\end{align}
And as in the Gaussian case, this indeed agrees with the leading term as implied by our corresponding formulas (\ref{F2}) and (\ref{20.2}).

The large deviation  tails of the PDF for the distribution of the largest eigenvalue in the $\beta=2, a=0$ Laguerre ensemble have been the subject of a recent experimental study \cite{FPNFD11}. We note that this particular Laguerre ensemble gives the eigenvalue PDF for the matrix product $X^{\dag} X$, where $X$ is an $N\times N$ standard complex Gaussian matrix (see \cite[\S 3.2]{Fo10}). The experimental study relates to the probability distribution of the combined output power from a certain coupled array of high gain lasers. It turns out that the round trip propagation matrix $M$ determining the evolution of the complex electric field in this setting can be decomposed in the form $X^{\dag} X$ for $X$ a square random matrix with near Gaussian complex entries. Moreover, the coupled lasers will lase at the mode with minimal losses, which corresponds to the eigenvector of $M$ with largest eigenvalue, and the output power is proportional to this eigenvalue. Significantly, very large data sets can be generated (order half a million measurements) allowing for the large deviation regime of the PDF to be probed.

In regards to the right tail, the experimental data was compared against the leading form (\ref{56}), modified by multiplication by a numerical factor $c_2 = 0.0063$ chosen on the basis of the modification then giving a better fit of results obtained from numerical simulation. But our result (\ref{F1}) extends (\ref{56}) to give the algebraic and constant terms. Explicitly, for $\beta=2, a=0$ we have
\begin{align}
\label{57} p_{N}^{\rm L} (4Ns) \sim e^{-4N (\sqrt{s(s-1)} + \log (\sqrt{s} -\sqrt{s-1}))} \frac{1}{s(s-1)} \left( \frac{1} {32\pi N}\right).
\end{align}
It thus remains to see if the right tail of the experimentally determined PDF can be resolved to the extent that this correction term can be exhibited.

The formulas (\ref{F1}) and (\ref{F2}) are valid for $N\to \infty$ with $s>1$ fixed. This is to be contrasted to the soft edge scaling limit in which $N\to \infty$ and simultaneously $s\to 1$ according to the scalings \cite{Fo93a}
\begin{align}
\label{S1} \sqrt{2N} s &= \sqrt{2N} + \frac{X} {\sqrt{2} N^{1/6}},\\
\label{S2} 4Ns &= 4N + 2(2N)^{1/3} X
\end{align}
in the Gaussian and Laguerre cases respectively. For $\beta$ even the explicit soft edge scaling form of the density $\rho_{(1)}^{\rm soft} (X)$ --- which as a consequence of the principle of universality \cite{Ku11} is independent of the particular case, Gaussian or Laguerre --- has been calculated in \cite[Corollary 12]{DF06}. Of interest to us is the corresponding right tail asymptotics deduced from this exact expression \cite[Corollary 14]{DF06} (corrected by the removal of a spurious factor of 1/2 on the RHS)
\begin{align}
\label{S3} \rho_{(1)}^{\rm soft} (X) \mathop{\sim}\limits_{X\to\infty} \frac{1}{\pi} \frac{\Gamma (1+\beta/2)}{(4\beta)^{\beta/2}} \frac{e^{-2\beta X^{3/2}/3}} {X^{3\beta/4 -1/2}}.
\end{align}
In the context of largest eigenvalue distributions, the argument leading to (\ref{3.2}) tells us that the PDF, $p^{\rm soft}(X)$ say, of the soft edge scaled largest eigenvalue of the Gaussian and Laguerre $\beta$-ensembles must have the same leading large $X$ asymptotic form as $\rho_{(1)}^{\rm soft}(X)$. Thus for $\beta$ even (at least)
\begin{align}
\label{S4} p^{\rm soft} (X) \mathop{\sim}\limits_{X\to\infty} \frac{1}{\pi} \frac{\Gamma (1+\beta/2)}{(4\beta)^{\beta/2}} \frac{e^{-2\beta X^{3/2}/3}} {X^{3\beta/4 -1/2}}.
\end{align}
The cases $\beta = 1$ and 2 of this expression, deduced using the exact Painlev\'e transcendent evaluations of $p^{\rm soft} (X)$ in these cases
(see e.g.~\cite[Ch.~9]{Fo10}) were reported in \cite{VBAB01}.
The equivalent asymptotic expression
\begin{align}
\label{S5} \int_X^{\infty} p^{\rm soft} (y) \, dy \mathop{\sim}\limits_{X\to\infty} \frac{1}{\pi \beta} \frac{\Gamma (1+\beta/2)}{(4\beta)^{\beta/2}} \frac{e^{-2\beta X^{3/2}/3}} {X^{3\beta/4}},
\end{align}
up to the explicit form of the constant but applying for general $\beta>0$, has recently been derived in \cite{DV11}, using the stochastic differential equation characterization.

It is possible to use our explicit large deviation forms (\ref{wf1}) and (\ref{20.2}) for $\rho_{(1),N}^{\rm G} (\sqrt{2N}s)$ and $\rho_{(1),N}^{\rm L} (4Ns)$ to reclaim (\ref{S3}). Thus as done in \cite{MV09} for the leading asymptotic forms (\ref{54}) and (\ref{56}), we replace $s$ by (\ref{S1}) (Gaussian case) and (\ref{S2}) (Laguerre case). Taking the limit $N\to \infty$ then reclaims the RHS of (\ref{S3}). In fact such correction formulas to the $X\to\infty$ form of $\rho_{(1)}^{\rm soft} (X)$ should hold for the complete asymptotic series of $\rho_{(1),N}^{\rm G} (\sqrt{2N} s)$,  $\rho_{(1),N}^{\rm L} (4N s)$ upon this double scaling.

The results of our study also have consequences with regard to the large deviation right tail form of the probability $E_{N, \beta} (k; (s,\infty))$ that there are precisely $k$ eigenvalues in the interval $(s, \infty)$ of the Gaussian and Laguerre $\beta$-ensembles. Now analogous to (\ref{3.1})
\begin{align}
\nonumber E_{N, \beta} (k; (s,\infty)) = &\int_s^{\infty} dt_1 \cdots \int_s^{\infty} dt_k  \, \rho_{(k)}(t_1,\dots,t_k)\\
\nonumber & - {k+1 \choose 1} \int_s^{\infty} dt_1 \cdots \int_s^{\infty} dt_{k+1} \, \rho_{(k+1)}(t_1,\dots,t_{k+1}) +\dots
\end{align}
For the Gaussian case, changing variables as in (\ref{3.1}) and making use of the analogue of (\ref{3.1e}) shows that for $s>1$
\begin{align}
\nonumber E_{N,\beta}^{\rm G} (k; (\sqrt{2N} s,\infty)) \sim (\sqrt{2N})^k \left(\int_s^{\infty} \rho_{(1),N}^G (\sqrt{2N} t) \, dt\right)^k
\end{align}
with $\sqrt{2N} \rho_{(1),N}^{\rm G} (\sqrt{2N} t)$ in turn given by (\ref{wf1}). Similarly, in the Laguerre case, we have that for $s>1$
\begin{align}
\nonumber E_{N,\beta}^{\rm L} (k; (4Ns,\infty)) \sim (4N)^k \left( \int_s^{\infty} \rho_{(1),N}^L (4N t) dt \right)^k
\end{align}
with $4N\rho_{(1),N}^{\rm L} (4Nt)$ given by (\ref{20.2}). The soft edge scaling limit of either of these reads
\begin{equation}\label{As}
E_\beta^{\rm soft}(k;(s,\infty)) \mathop{\sim}\limits_{s \to \infty} \Big ( \int_s^\infty \rho_{(1)}^{\rm soft}(t) \, dt \Big )^k.
\end{equation}
This asymptotic form compliments that for $E_\beta^{\rm soft}(k;(s,\infty))$ with $s \to - \infty$ recently derived in \cite{FW11}.

\subsection*{Acknowledgements}
The participation of Chris Ormerod in the early stages of this project, and the help
of Wendy Baratta and Anthony Mays in the preparation of the manuscript, is
acknowledged. This work was supported by the Australian Research Council.

\appendix
\section{Appendix}
\setcounter{equation}{0}

Here the results (\ref{wf}) and (\ref{20.1}) for $\beta$ even will be derived by applying the steepest descent/ stationary phase method of asymptotic analysis to duality formulas giving the averages as $\beta$-dimensional integrals.

\subsection{Gaussian case}

According to \cite[(13.162)]{Fo10}
\begin{align}
\label{A.1} \left\langle \prod_{j=1}^N (c-\sqrt{\alpha} \: y_j)^n \right\rangle_{{\rm ME}_{2/\alpha,N}(e^{-y^2})}= \left\langle \prod_{j=1}^n (c-iy)^N \right\rangle_{{\rm ME}_{2\alpha,n}(e^{-y^2})}.
\end{align}
Let $\beta$ be even and write $1/\alpha = \beta/2$, $n=\beta$. After scaling the variables (\ref{A.1}) then reads
\begin{align}
\label{A.2} \left\langle \prod_{j=1}^N (x-y_j)^{\beta} \right\rangle_{{\rm ME}_{\beta,N}(e^{-\beta My^2})} = \left\langle \prod_{j=1}^\beta (x-iy_j)^{N} \right\rangle_{{\rm ME}_{4/\beta,\beta}(e^{-2 My^2})}.
\end{align}
In words (\ref{A.2}) 
tells us that the $\beta$ moment of the characteristic polynomial in ${\rm ME}_{\beta,N} (e^{-\beta M y^2})$ is proportional to the $N$-th moment of the characteristic polynomial at $-ix$ in ${\rm ME}_{ 4/\beta,\beta} (e^{-2My^2})$. The latter is well suited to an $N\to \infty$ asymptotic analysis.

As a multidimensional integral the RHS of (\ref{A.2}) reads
\begin{align}
\nonumber &\left\langle \prod_{j=1}^N (x-iy_j)^N \right\rangle_{{\rm ME}_{4/\beta,\beta}(e^{-2My^2})} = \frac{1}{C_{4/\beta,\beta} [e^{-2My^2}]}\\
\label{A.3}&   \qquad  \times \int_{-\infty}^{\infty} dy_1 \cdot\cdot\cdot \int_{-\infty}^{\infty} dy_{\beta} \, e^{-2(M-N)\sum_{j=1}^{\beta} y_j^2} e^{N\sum_{j=1}^{\beta} f(x,y_j)} \prod_{1\leq < k\leq \beta} |y_j - y_k|^{4/\beta},
\end{align}
where
\begin{align}
f(x,y):= -2y^2 +\log (x-iy).
\end{align}
For $N\to \infty$ this can be analyzed using the method of stationary phase

To apply this method, we note
\begin{align}
\label{f1} \frac{\partial f}{\partial y} &= -4y - \frac{i}{x-iy},\\
\label{f2} \frac{\partial^2 f}{\partial y^2} &= -4 + \frac{1}{(x-iy)^2}.
\end{align}
The first of these implies that for $x>1$ the stationary points of $f$ occur at
\begin{align}\label{f3}
y_{\pm} = -\frac{i}{2} (x \pm (x^2 -1)^{1/2}).
\end{align}
We would like to deform the contours of integration from the real line as in (\ref{A.3}) to pass through one of these points.
The details of how to rewrite (\ref{A.3}) so the integrand is an analytic function is given in \cite[Lemma 1]{DF06}. Noting from (\ref{f2}) that $\partial^2 f/\partial y^2\big|_{y=y_{-}}< 0$ suggests we translate each contour parallel to the real axis to pass through $y_{-}$ (the other choice $y=y_{+}$ is such that $\partial^2 f/\partial y^2 >0$, which would mean deforming the contour to run parallel to the imaginary axis in the complex plane, causing convergence problems at infinity). Furthermore, we expand 
\begin{align*}
\nonumber f(x,y_j) = f(x,y_{-}) - \frac{1}{2}\alpha_x (y_{-}-y_j)^2 + \cdot\cdot\cdot,
\end{align*}
where
\begin{align}
\label{f4} f(x,y_{-}) &= x^2 - x(x^2-1)^{1/2}-\frac{1}{2} +\log \left( \frac{x+(x^2-1)^{1/2}}{2}\right),\\
\label{f5} \alpha_x &= 8(x^2-1)^{1/2} (x-(x^2-1)^{1/2}).
\end{align}
Substituting in (\ref{A.3}) and changing variables shows, after setting $M=N+1$
\begin{align}
\left\langle \prod_{j=1}^N (x-iy_j)^N \right\rangle_{{\rm ME}_{\beta,4/\beta} (e^{-2(N+1)y^2})} \mathop{\sim}\limits_{N\to\infty} e^{-2\beta y^2_{-} +N\beta f(x,y_{-})} \left( \sqrt{\frac{4} {\alpha_x}}\right)^{3\beta -2}.
\end{align}
Substituting (\ref{f3})--(\ref{f5}) reclaims (\ref{wf}).

\subsection{Laguerre case}
\setcounter{equation}{0}

The analogue of (\ref{A.2}) in the Laguerre case is the identity \cite[eq. (13.44), Exercises 13.1, q.4]{Fo10}
\begin{align}
\nonumber &\frac{C_{N,\beta}\left[ x^{\beta a/2} e^{-\beta x/2}\right] (4M)^{Nm}} {C_{N,\beta} [x^{\beta a/2 +m} e^{-\beta x/2}]} \left\langle \prod_{j=1}^N (x_j-t)^m \right\rangle_{{\rm ME}_{N,\beta} (x^{\beta a/2} e^{-2M \beta x})}\\
\nonumber &\qquad =\frac{1}{M_m (\hat{a},N,2/\beta)} \oint \frac{dz_1}{2\pi i} \cdot\cdot\cdot \oint \frac{dz_m}{2\pi i} \prod_{l=1}^m e^{-4Mtz_l} z_l^{-N-1-(2/\beta) (m-1)} (1+z_l)^{\hat{a}+N}\\
\label{L1} &\qquad \quad \times \prod_{1\leq < k \leq m} (z_k - z_j)^{4/\beta},
\end{align}
where $\hat{a} := a-1+2/\beta$ and
\begin{align}
\label{L1a} M_m(\alpha, \beta, \gamma) = \prod_{j=0}^{m-1} \frac{\Gamma (1+\alpha +\beta +j\gamma) \Gamma (1+(j+1)\gamma)} {\Gamma (1+\alpha +j\gamma) \Gamma (1+\beta + j\gamma) \Gamma (1+\gamma)}.
\end{align}
With $m=\beta$ ($\beta$ even) this expresses the $\beta$ moment as a $\beta$-dimensional integral, and furthermore the latter is well suited to asymptotic analysis. But before doing so, we apply Stirling's formula to the known gamma function evaluation of the normalization on the LHS of (\ref{L1}) as used in the derivation of (\ref{9ag}) to obtain their asymptotic form. This allows us to replace the LHS of (\ref{L1}) by
\begin{align}
\label{L2f} \frac{e^{M\beta} 2^{2N\beta} \Gamma (a\beta/2 +1) \Gamma ((a+1)\beta/2)}{(\pi N \beta) \left( N\beta /2\right)^{(2a+1)\beta/2}} \left\langle \prod_{j=1}^N (x_j-t)^\beta \right\rangle_{{\rm ME}_{N,\beta} (x^{\beta a/2} e^{-2M \beta x})}.
\end{align}
We also simplify the prefactor on the RHS. Thus it follows  from (\ref{L1a}), the duplication formula for the gamma function
\begin{align}
\label{L2x} \prod_{j=0}^{\beta/2-1} \Gamma (z+ 2j/\beta) = (2\pi)^{(\beta/2 -1)/2} (\beta/2)^{1/2 -z\beta/2} \Gamma (\beta z/2)
\end{align}
and Stirling's formula that for $N\to\infty$
\begin{align}
\label{L3g} \frac{1}{M_{\beta} (\hat{a},N,2/\beta)} \sim \left( \Gamma (1+2/\beta)\right)^{\beta} (\beta/2)^2 (\beta N /2)^{\beta (1-a)-2} \frac{\Gamma (\beta a/2+1) \Gamma( \beta(a+1)/2 +1)}{\Gamma (\beta/2 +1) \Gamma (\beta +1)}.
\end{align}

We now turn our attention to the integral on the RHS of (\ref{L1}), and for convenience specialize to the case $M=N+1$, although our final formula will be valid for any $M-N$ fixed. This integral can be written
\begin{align}
\label{L3h} \oint \frac{dz_1}{2\pi i} \cdot\cdot\cdot \oint \frac{dz_{\beta}}{2\pi i} \prod_{l=1}^{\beta} e^{Mg(t,z_l)} (1+z_l)^{a-2+2/\beta} z_l^{-(2/\beta)(\beta-1)} \prod_{1\leq j< k\leq \beta} (z_k-z_l)^{4/\beta},
\end{align}
where
\begin{align}
g(t,z):= -4tz- \log z +\log (1+z).
\end{align}

In preparation for applying the method of stationary phase, we note that
\begin{align}
\nonumber \frac{\partial g}{\partial z} &=-4t -\frac{1}{z} +\frac{1}{1+z}\\
\nonumber \frac{\partial^2 g}{\partial z^2} &= \frac{1}{z^2} - \frac{1}{(1+z)^2}.
\end{align}
Thus there are stationary points at
\begin{align}
\nonumber z_{\pm} = \frac{1}{2} (-1\pm (1-1/t)^{1/2})
\end{align}
and at these points
\begin{align}
\label{L4} \frac{\partial^2 g}{\partial z^2} &= \pm \gamma_t, \qquad \gamma_t := (4t)^2 (1-1/t)^{1/2}.
\end{align}
Choosing the positive sign allows the contours (which must all encircle the origin) to be deformed to pass through $z_{+}$ parallel to the imaginary axis; the choice $z_{-}$ would require deforming the contour along the (negative) real axis, which is a branch cut of the logarithm function. Expanding about $z_{+}$ we have to second order
\begin{align}
\nonumber g(t,z_j) = g(t,z_{+}) -\frac{1}{2} \gamma_t \: x_j^2, \qquad z_j-z_{+} =ix_j,
\end{align}
with
\begin{align}
\nonumber g(t,z_{+}) = 2t(1- (1-1/t)^{1/2}) - 2\log (\sqrt{t}-\sqrt{t-1})
\end{align}
and $\gamma_t$ as in (\ref{L4}). Substituting in (\ref{L3h}) and making a further change of variables $\sqrt{\gamma_t} x_j = y_j$ shows that for $N\to\infty$ (\ref{L3h}) has the asymptotic form
\begin{align}
\nonumber &\left( \frac{1}{2\pi} \right)^{\beta} C_{\beta, 4/\beta} [e^{-x^2/2}] (1+z_{+})^{\beta (a-2+2/\beta)} z_{+}^{-2(\beta-1)} (N\gamma_t)^{(2-3\beta)/2}\\
\label{L5} & \quad = \frac{\Gamma (1+\beta/2) \Gamma (1+\beta)} {2\pi (\Gamma (1+ 2/\beta))^{\beta}} (\beta/2)^{1-(\beta/2)(3+4/\beta)} (1+z_{+})^{\beta (a-2+2/\beta)} z_{+}^{-2(\beta-1)} (N\gamma_t)^{(2-3\beta)/2},
\end{align}
where the equality follows upon making use of the explicit gamma function evaluation of $C_{\beta,4/\beta} [e^{-x^2/2}]$ \cite[eq. (4.140)]{Fo10} and the duplication formula (\ref{L2x}).

Multiplying (\ref{L5}) and (\ref{L3g}), equating to (\ref{L2f}) and simplifying reclaims (\ref{20.1}).


\begin{thebibliography}{10}

\bibitem{BDG01}
G.~Ben Arous, A.~Dembo, and A.~Guionnet, \emph{Aging in spherical spin
  glasses}, Prob. Th. Rel. Fields \textbf{120} (2001), 1--67.

\bibitem{BF97a}
T.H. Baker and P.J. Forrester, \emph{The {Calogero-Sutherland} model and
  generalized classical polynomials}, Commun. Math. Phys. \textbf{188} (1997),
  175--216.

\bibitem{BG11}
G.~Borot and A.~Guionnet, \emph{Asymptotic expansion of beta matrix models in
  the one-cut regime}, arXiv:1107.1167, 2011.

\bibitem{DF06}
P.~Desrosiers and P.J. Forrester, \emph{Hermite and {L}aguerre
  $\beta$-ensembles: asymptotic corrections to the eigenvalue density}, Nucl.
  Phys. B \textbf{743} (2006), 307--332.

\bibitem{DV11}
L.~Dumaz and B.~Vir\'ag, \emph{The right tail exponent of the
  {T}racy-{W}idom-beta distribution}, arXiv:1102.4818, 2011.

\bibitem{DE05}
I.~Dumitriu and A.~Edelman, \emph{Global spectrum fluctuations for the
  $\beta$-{H}ermite and $\beta$-{L}aguerre ensembles via matrix models}, J.
  Math. Phys. \textbf{47} (2006), 063302(2006).

\bibitem{EKS94}
A.~Edelman, E.~Kostlan, and M.~Shub, \emph{How many eigenvalues of random
  matrix are real?}, J. Amer. Math. Soc. \textbf{7} (1994), 247--267.

\bibitem{ES06}
A.~Edelman and B.D. Sutton, \emph{From random matrices to stochastic
  operators}, J. Stat. Phys. \textbf{127} (2006), 1121--1165.

\bibitem{Fo93a}
P.J. Forrester, \emph{The spectrum edge of random matrix ensembles}, Nucl.
  Phys. B \textbf{402} (1993), 709--728.

\bibitem{Fo10}
\bysame, \emph{Log-gases and random matrices}, Princeton University Press,
  Princeton, NJ, 2010.

\bibitem{FF03}
P.J. Forrester and N.E. Frankel, \emph{Applications and generalizations of
  {F}isher-{H}artwig asymptotics}, J. Math. Phys. \textbf{45} (2003),
  2003--2028.

\bibitem{FFG06e}
P.J. Forrester, N.E. Frankel, and T.M. Garoni, \emph{Asymptotic form of the
  density profile for {G}aussian and {L}aguerre random matrix ensembles with
  orthogonal and symplectic symmetry}, J. Math. Phys. \textbf{47} (2006),
  023301.
  
  \bibitem{FW11}
  P.J.~Forrester and N.S.~Witte, \emph{Asymptotic forms for hard and soft edge general
  $\beta$ conditional gap probabilities}, arXiv:1110.4284, 2011.


\bibitem{FPNFD11}
M.~Fridman, R.~Pugatch, M.~Nixon, A.A. Friesem, and N.~Davidson,
  \emph{Measuring maximal eigenvalue distribution of {W}ishart random matrices
  with coupled lasers}, arXiv:1012.1282, 2010.

\bibitem{Fy04}
Y.V. Fyodorov, \emph{Complexity of random energy landscapes, glass transition
  and absolute value of spectral determinant of random matrices}, Phys. Rev.
  Lett. \textbf{92} (2004), 240601, Erratum: Phys. Rev. Lett. 93 (2004),
  149901.
  

\bibitem{Jo98}
K.~Johansson, \emph{On fluctuation of eigenvalues of random {Hermitian}
  matrices}, Duke Math. J. \textbf{91} (1998), 151--204.

\bibitem{K07}
I.~Kravotsky, \emph{Correlations of the characteristic polynomials in the
  {G}aussian unitary ensemble or a singular {H}ankel determinant}, Duke Math.
  J. \textbf{139} (2007), 581--619.

\bibitem{Ku11}
A.~Kuijlaars, \emph{Universality}, The {O}xford {H}andbook of {R}andom {M}atrix
  {T}heory (G.~Akemann, J.~Baik, and P.~di~Francesco, eds.), Oxford University
  Press, Oxford, 2011, pp.~103--134.

\bibitem{MV09}
S.N. Majumdar and M.~Vergassola, \emph{Extreme value statistics of eigenvalues
  of {G}aussian random matrices}, Phys. Rev. Lett. \textbf{102} (2009), 060601.

\bibitem{Me91}
M.L. Mehta, \emph{Random matrices}, 2nd ed., Academic Press, New York, 1991.

\bibitem{NM11}
C.~Nadal and S.N. Majumdar, \emph{A simple derivation of the {T}racy-{W}idom
  distribution of the maximal eigenvalue of a {G}aussian unitary random
  matrix}, J. Stat. Mech. \textbf{2011} (2011), P04001.

\bibitem{VBAB01}
M.G. Vavilov, P.W. Brouwer, V.~Ambegaokar, and C.W.J. Beenakker,
  \emph{Universal gap fluctuations in the superconductor proximity effect},
  Phys. Rev. Lett. \textbf{86} (2001), 874--877.

\end{thebibliography}

\providecommand{\bysame}{\leavevmode\hbox to3em{\hrulefill}\thinspace}
\providecommand{\MR}{\relax\ifhmode\unskip\space\fi MR }
\providecommand{\MRhref}[2]{%
  \href{http://www.ams.org/mathscinet-getitem?mr=#1}{#2}
}
\providecommand{\href}[2]{#2}

\end{document}